\documentclass[a4paper,UKenglish,cleveref, autoref, thm-restate]{lipics-v2021}

\pdfoutput=1 
\hideLIPIcs  

\usepackage{todonotes}
\usepackage{algorithm}
\usepackage{algpseudocode}

\usepackage{rotating}
\usepackage{booktabs}

\newcommand{\End}{\mathtt{e}}
\newcommand{\Begin}{\mathtt{b}}

\newcommand{\Cand}{\mathtt{Candt}}
\newcommand{\FullCand}{\mathtt{FullL}}
\newcommand{\Weight}{\mathtt{w}}
\newcommand{\RowList}{\mathtt{rowList}}
\newcommand{\Result}{\mathtt{result}}
\newcommand{\Map}{\mathtt{MAP}}

\newcommand{\SA}{\mathtt{SA}}
\newcommand{\LCP}{\mathtt{LCP}}
\newcommand{\BWT}{\mathtt{BWT}}
\newcommand{\posprevWeight}{\mathtt{prevW}}
\newcommand{\POS}{\mathtt{POS}}

\newcommand{\CurrB}{\mathtt{curr\_b}}

\title{Constructing Suffixient Arrays Revisited} 

\author{Paola Bonizzoni}{Department of Computer Science, University of Milano-Bicocca, Italy}{paola.bonizzoni@unimib.it}{https://orcid.org/0000-0001-7289-4988}{}

\author{Younan Gao}{Department of Computer Science, University of Milano-Bicocca, Italy}{younan.gao@unimib.it}{https://orcid.org/0000-0003-4984-2551}{}

\author{Brian Riccardi}{Department of Computer Science, University of Milano-Bicocca, Italy}{brian.riccardi@unimib.it}{https://orcid.org/0000-0002-4925-9529}{}

\authorrunning{Bonizzoni et al.} 

\Copyright{Paola Bonizzoni, Younan Gao, and Brian Riccardi} 



\category{} 

\relatedversion{} 

\funding{All authors have received funding from the European Union’s Horizon 2020 research and innovation programme under the Marie Sk{\l}odowska-Curie grant agreement PANGAIA No.~872539, as well as from grant MIUR 2022YRB97K (PINC, \emph{Pangenome Informatics: From Theory to Applications}), funded by the European Union under the NextGenerationEU programme, Mission~4.}

\acknowledgements{}

\nolinenumbers 

\EventEditors{Philip Bille and Nicola Prezza}
\EventNoEds{2}
\EventLongTitle{37th Annual Symposium on Combinatorial Pattern Matching (CPM 2026)}
\EventShortTitle{CPM 2026}
\EventAcronym{CPM}
\EventYear{2026}
\EventDate{June 15--17, 2026}
\EventLocation{Copenhagen, Denmark}
\EventLogo{}
\SeriesVolume{369}
\ArticleNo{33}

\begin{document}

\maketitle

\begin{abstract}
Recently, Cenzato et al.\ proposed a new text index, called the \emph{suffixient array}, which is a subset of the suffix array and supports locating a single pattern occurrence or finding its maximal exact matches (MEMs), assuming random access to the input text $T[1..n]$ is available.
They show that, given the suffix array, the longest common prefix array, and the Burrows--Wheeler transform ($\BWT$) of the reverse of $T[1..n]$ over an alphabet $\{1,\ldots,\sigma\}$, a suffixient array can be constructed in linear time.
However, their construction algorithms require multiple scans of these arrays.
When restricted to a single pass over the arrays, they present an alternative construction algorithm running in $\mathcal O(n + \overline{r} \log \sigma)$ time, where $\overline{r}$ is the number of runs in the $\BWT$ of the reversed text.
In this paper, we present a new one-pass algorithm that constructs a suffixient array in linear time under the standard RAM model.
\end{abstract}

\newpage
\section{Introduction}
\label{sec:intro}

Pattern matching is a fundamental problem with applications ranging from text processing to computational biology. The \emph{suffix array}~\cite{ManberMyers1993} is a textbook data structure that supports efficient pattern matching by storing the starting positions of all suffixes of a text in lexicographical order. Together with auxiliary structures such as the \emph{longest common prefix array}~\cite{ManberMyers1993}, the suffix array enables fast pattern searches and serves as the basis of many full-text indexes.
However, the limitations of suffix arrays become apparent when dealing with massive texts, such as collections of human genomes. Since a suffix array requires linear space in the text length, storing and processing it for such datasets is often infeasible in practice.

At the same time, genomic data exhibits a high degree of similarity and repetitiveness: genomes from different individuals share the vast majority of their sequences, with variations occurring only at relatively sparse locations. This strong repetitiveness suggests that designing space-efficient indexes whose space usage scales with intrinsic measures of repetitiveness, rather than with the raw text length~$n$, is both natural and highly desirable.

\emph{Suffixient arrays}~\cite{cenzato2025} were recently proposed as a space-efficient alternative to suffix arrays for indexing highly similar texts such as genome sequences. Intuitively, a suffixient array can be viewed as a carefully selected subset of the \emph{prefix array}, the symmetric counterpart of the suffix array. Its definition is grounded in the notion of \emph{right-maximal substrings}, that is, substrings that occur in the text and can be extended to the right by at least two distinct characters. Such substrings correspond exactly to the internal nodes of the suffix tree. The suffixient array captures these branching points by including prefixes of the text that cover all one-character right-maximal extensions, ensuring that the structure is sufficient to represent all essential distinctions induced by right-maximal extensions.

Despite its reduced size, a suffixient array supports a specialized set of queries, provided random access to the underlying text: it can be used to either locate a single occurrence of a given pattern or to find \emph{maximal exact matches} (MEMs), which are of particular importance in bioinformatics applications~\cite{cenzato2025}.

In this paper, we focus on the efficient construction of suffixient arrays. Unlike suffix arrays, which are uniquely determined for a given text, a text may admit multiple valid suffixient arrays. Our goal is to compute any one of them.

\subparagraph*{Related Work.}  
Recently, Navarro et~al.~\cite{navarro2025smallest} introduced an \emph{online construction algorithm} that computes a \emph{minimum-size suffixient set}---the set of indices forming a suffixient array---directly from the text. Their approach adapts Ukkonen’s algorithm~\cite{ukkonen1995line}, originally designed for the online construction of the suffix \emph{tree}. Specifically, their method incrementally maintains a minimum-size suffixient set for the prefix $T[1..i]$ as the text is scanned left to right.  

The algorithm runs in linear time over an alphabet of size $\sigma$ under the trans-dichotomous RAM model~\cite{fredman1993surpassing}, but requires $\mathcal O(n \log \sigma)$ time in the standard RAM model~\cite{aho1974design}. Its working space is comparable to that of a suffix tree built over the text, which is known to have a large memory footprint in practice---up to 20 bytes per input character in the worst case~\cite{kurtz1999reducing}. Navarro et~al.\ also demonstrated that a minimum-size suffixient set can be constructed incrementally while scanning the text from right to left.


Cenzato et~al.~\cite{cenzato2025} addressed the problem of constructing a minimum-size suffixient set given the suffix array $\SA[1..n]$, the Burrows--Wheeler transform (BWT)~\cite{BurrowsWheeler1994} array $\BWT[1..n]$, and the longest common prefix array $\LCP[1..n]$ of the reversed text. They presented two linear-time algorithms that require multiple passes over these arrays, resulting in a working space of $\Omega(n)$ words.
They also proposed a construction algorithm that avoids storing the full $\SA$, $\BWT$, and $\LCP$ arrays and processes them in a single pass. This algorithm runs in $\mathcal O(n + \overline{r}\log \sigma)$ time, where $\overline{r}$ denotes the number of runs in the BWT, and uses only $\mathcal O(\sigma)$ words of working space. When computing a suffixient array, the running time remains unchanged, while the working space increases from $\mathcal O(\sigma)$ to $\mathcal O(\sigma + \chi) \subseteq \mathcal O(\chi)$ words, where $\chi$ denotes the size of the suffixient array and $\sigma \le \chi$~\cite{cenzato2025}.

Using the technique of \emph{prefix-free parsing} (PFP)~\cite{Boucher2019PFP, kuhnle2020efficient, rossi2022moni}, the entries of the $\SA$, $\BWT$, and $\LCP$ arrays can be generated in a left-to-right streaming fashion while using compressed space. By combining their one-pass algorithm with PFP, Cenzato et~al.~\cite{cenzato2025} were able to construct a suffixient array in compressed space.

\subparagraph*{Our Results.}
In this paper, we revisit the problem of constructing a suffixient array in the standard RAM model, focusing on the specific setting introduced by Cenzato et~al.~\cite{cenzato2025}.
Our main contribution is a linear-time streaming algorithm that constructs a suffixient array in $n$ iterations by processing the $\SA$, $\LCP$, and $\BWT$ arrays of the reversed text $T[1..n]$ in a single pass with one-step look-ahead.
The efficiency of the algorithm relies on two key mechanisms.

\textit{Dynamic Candidate Management.}
At each iteration $i$, we maintain a set of candidate positions $j_{\text{text}}$ for the suffixient array in a doubly-linked list ($\RowList$).
The list has size at most $\sigma$ and is kept sorted in non-increasing order by the weight $\Weight(j)$.
Here, $\Weight(j)$ denotes the maximum length of a right-maximal substring that is a suffix of $T[1..j_{\text{text}}-1]$, where $j = \SA^{-1}[n-j_{\text{text}}+1]$. The array $\SA^{-1}$ denotes the inverse suffix array, defined by $\SA^{-1}[\SA[k]] = k$.

To decide at iteration $i$ whether the new position $i_{\text{text}}=n-\SA[i]+1$ should replace or join the existing candidates, we use an auxiliary array $\posprevWeight[1..\sigma]$ that records, for each character, the most recent weighted occurrence (its $\mathit{index}$) together with the corresponding $\mathit{weight}$.
In addition, we employ a monotone stack to compute in amortized $\mathcal O(1)$ time $\Begin(i)$, the largest index $i'<i$ such that $\LCP[i']<\LCP[i]$.
By comparing $\Weight(i)$ with $\LCP[i]$, $\posprevWeight[c].\mathit{index}$ with $\Begin(i)$, and $\posprevWeight[c].\mathit{weight}$ with $\Weight(i)$ (where $c=\BWT[i]$), we can determine whether the position $i_{\text{text}}$ is a superior candidate to those currently stored.
This strategy guarantees that, at any time, we retain only the most relevant candidate for each distinct character. \lipicsEnd

\textit{The Ejection Mechanism.}
As the algorithm streams through the arrays, the current $\LCP$ value, $\LCP[i]$, acts as a threshold.
When $\LCP[i]$ drops below the weight of a candidate in $\RowList$, this indicates that the corresponding right-maximal extension can no longer be extended.
The candidate \emph{position} is then \emph{ejected} from the doubly-linked list and appended to the output list, which stores a sublist of the suffixient array.
Since $\RowList$ is sorted, each ejection takes constant time, resulting in overall linear running time. \lipicsEnd

Our algorithm is asymptotically faster than the one-pass algorithm of Cenzato et~al.~\cite{cenzato2025}. Its working space is bounded by $\mathcal O(h + \chi)$, in addition to the $\SA$, $\LCP$, and $\BWT$ arrays, while computing a minimum-size suffixient set requires only $\mathcal O(h + \sigma)$ words of working space.
Here, $h$ denotes the height of the suffix tree built over the reversed text, that is, the maximum number of branching nodes along any root-to-leaf path. 
Since the algorithm processes these arrays in a one-pass streaming fashion, it naturally integrates with prefix-free parsing~\cite{Boucher2019PFP}.

\section{Preliminaries}

Our model of computation is a random access machine (RAM) endowed with comparison operations and basic arithmetic operations, including only addition and subtraction \cite{aho1974design}.

Let $\Sigma=\{1, \dots, \sigma\}$ be the alphabet set.
We assume that the text $T[1..n]$ ends at the position $n$ with a special character $\$$ that never appears in $T[1..n-1]$.
So, the alphabet size is at least two.
Without loss of generality, we further assume that $\sigma\le n$.

For any string $S$, let $(S)^{rev}$ denote its reverse. We define the reversed text $T^{rev}[1..n]$ as a specific construction distinct from $(T)^{rev}$; while $(T)^{rev}$ simply reverses the entire sequence, $T^{rev}$ preserves the sentinel $\$$ at the final position.
Formally, $T^{rev}[n] = \$$ and $T^{rev}[i] = T[n-i]$ for $1 \le i < n$.
Let $\SA[1..n]$, $\BWT[1..n]$, and $\LCP[1..n]$ denote the suffix array, the Burrows--Wheeler transform, and the Longest Common Prefix array of the reversed text $T^{rev}$, respectively.
Specifically, $\SA[1..n]$ is a permutation of the indices $\{1, \dots, n\}$ such that the suffixes $T^{rev}[\SA[i]..n]$ are arranged in lexicographical order for $1 \le i \le n$.
And $\BWT[i] = \$ $ if $\SA[i] = 1$; otherwise, $\BWT[i] = T^{rev}[\SA[i]-1] = T[n-\SA[i]+1]$. Moreover, $\LCP[1] = -1$, and for $i > 1$, $\LCP[i]$ is the length of the longest common prefix of $(T[1..n-\SA[i]])^{rev}$ and $(T[1..n-\SA[i-1]])^{rev}$.

Given two strings $\alpha$ and $\beta$, we say that $\alpha$ is \emph{co-lexicographically smaller} than $\beta$ if and only if one of the following holds: i) there exists an index $k$ such that $\alpha[|\alpha|-i+1]=\beta[|\beta|-i+1]$ for all $1\le i <k$, and $\alpha[|\alpha|-k+1]<\beta[|\beta|-k+1]$; ii) or $\alpha$ is a proper suffix of $\beta$.

Throughout this paper, we reserve the term \emph{position} for locations in the text $T[1..n]$, and the term \emph{index} for locations in the arrays $\SA$, $\BWT$, and $\LCP$.
For example, we use the symbol $j_{\text{text}}$ to denote a position in $T[1..n]$, and write $j$ (without the subscript ``text'') to denote an index in these arrays.
Consequently, for $1 \le j \le n$ we have $j_{\text{text}} = n - \SA[j] + 1$, and for $1 \le j_{\text{text}} \le n$ we have $j = \SA^{-1}[\,n - j_{\text{text}} + 1\,]$.


\begin{definition}[Right-maximal substrings and one-character right-maximal extension~\cite{cenzato2025}]
For a text $T[1..n]$ containing at least two distinct characters, including $\$$, a substring $T[i_{\text{text}}.. j_{\text{text}}]$ $(j_{\text{text}}\ge i_{\text{text}}-1)$ is a \emph{right-maximal substring} if there exist at least two distinct characters $a, b \in \Sigma$ such that both $T[i_{\text{text}}.. j_{\text{text}}] \cdot a$ and $T[i_{\text{text}}.. j_{\text{text}}] \cdot b$ are substrings of $T$. 
For any right-maximal substring $str$, we call $str \cdot c$ for $c \in \Sigma$ a \emph{one-character right-maximal extension} of $str$ if $str \cdot c$ is a substring of $T[1..n]$.
\end{definition}

Note that the empty string is always a right-maximal substring, since $\sigma\ge 2$.

\begin{definition}[Suffixient set~\cite{cenzato2025}]\label{def-suffixient-set}
A set $\mathcal{X}\subseteq \{1,\dots, n\}$ is \emph{suffixient} for a text $T[1..n]$ if, for every one-character right-maximal extension $T[i_{\text{text}}.. j_{\text{text}}]$ $(j_{\text{text}}\ge i_{\text{text}})$ of every right-maximal string $T[i_{\text{text}}.. j_{\text{text}}-1]$, there exists a position $x_{\text{text}}\in \mathcal{X}$ such that $T[i_{\text{text}}.. j_{\text{text}}]$ is a suffix of $T[1.. x_{\text{text}}]$.
\end{definition}

Since $T[n] = \$$ and $\$$ does not appear in $T[1..n-1]$, any suffixient set must contain $n$.


\begin{definition}[Suffixient array~\cite{cenzato2025}]
A \emph{suffixient array} $sA$ of a text $T[1..n]$ is a minimum-size suffixient set for $T$, whose elements are ordered so that the prefixes $T[1..sA[1]], T[1..sA[2]], \dots$ appear in colexicographic order.
\end{definition}




\subparagraph*{Computing the \(\Weight(\cdot)\) values.}
Let $x \in [1..n]$ be an index. We say that $x$ is a \emph{start-run boundary} if $x > 1$ and $\BWT[x] \neq \BWT[x-1]$, and an \emph{end-run boundary} if $x < n$ and $\BWT[x] \neq \BWT[x+1]$. Any index that is either a start- or end-run boundary is called a \emph{run boundary}.

To every index $x \in [1..n]$ we assign a \emph{weight} \(\Weight(x)\), defined as the maximum length of a right-maximal substring that is a suffix of \(T[1..n - \SA[x]]\), if $x$ is a run-boundary, and $-1$, otherwise. Note that if $x$ is a
run-boundary, then a right-maximal substring ending at position $n-\SA[x]$ surely
exists, hence $\Weight(x)$ is well-defined.

The value \(\Weight(x)\) can then be computed as follows: if \(x\) is not a run boundary, then we set \(\Weight(x) := -1\), as per the definition. If \(x\) is only a start-run (resp., only an end-run) boundary, then \(\Weight(x) = \LCP[x]\) (resp., \(\Weight(x) = \LCP[x+1]\)). If \(x\) is both a start- and end-run boundary, then \(\Weight(x) = \max\{\LCP[x], \LCP[x+1]\}\).
Each value can be computed in $\mathcal O(1)$ time using $\mathcal O(1)$ working space. The pseudocode is deferred to Appendix~\ref{app-pseudo-weight}.

\subparagraph*{Computing the $\Begin(\cdot)$ values.} For any $x \in [1..n]$, define $\Begin({x})$ as the largest index $\Begin({x}) < {x}$ for which $\LCP[\Begin({x})] < \LCP[{x}]$. If no such index exists, set $\Begin({x}) := 1$. Define $\End({x})$ as the smallest index $\End({x}) > {x}$ such that $\LCP[\End({x})] < \LCP[{x}]$. If no such index exists, set $\End({x}) := n+1$.
The value    $\End({x})$ is only used conceptually throughout this paper, so we only show an algorithm that computes $\Begin({x})$.

In our setting, the entries in $\LCP[1..n]$ are enumerated sequentially from $\LCP[1], \LCP[2], \dots$ in a stream; immediately after reading the entry $\LCP[i]$ for each $i>1$, our goal is to compute the index $\Begin(i)$ in amortized constant time.

Our data structure is a regular stack that stores tuples of form $(index, val, b\_val)$, where $index \in [1..n]$, $val=\LCP[index]$, and $b\_val=\Begin(index)$.
We call this stack \emph{monotone stack} as the tuples in the stack are always sorted decreasingly by the $val$ entries from top to the bottom.
Initially, we create an empty stack $S$ and push the triple $(1, -1, -1)$ onto $S$.


When $\LCP[i]$ is available, we first check if or not the tuple at the top of $S$ has $index = i$. If so, then we return the $b\_val$ value of the top tuple; this means that $\Begin(i)$ has already been computed before (note that in our setting, when $\LCP[i]$ is available, we might call $\Begin(i)$ more than once). 
Otherwise, we pop every tuple in the stack with $val\ge \LCP[i]$ out of $S$, until we find the tuple $t$ with $t.val<\LCP[i]$.
Then, we return $t.index$ right after we push the tuple $(i, \LCP[i], t.index)$ onto the top of $S$.
The pseudocode can be found in Appendix~\ref{app-pseudo-begin}.

Since the tuple associated with each $\LCP[i]$ is pushed onto $S$ exactly once and popped from $S$ at most once, the overall running time is $\mathcal O(n)$. Consequently, $\Begin(i)$ can be computed in amortized constant time. The correctness is standard and omitted; a similar idea has appeared in \cite[Theorem~2]{fischer2009faster} and \cite[Section 4]{abouelhoda2004replacing} for a related purpose.

Let $|S|_i$ denote the size of the stack immediately after processing $\LCP[i]$.  By Proposition~\ref{prop} (below), $\max\{ |S|_i \mid 1 < i \le n \}$ is upper bounded by the height of the \emph{suffix tree} \cite{weiner1973linear} built over $T^{rev}[1..n]$, plus one. 
In particular, if $T[1..n]=a^n$, then $|S|\subseteq \Omega(n)$ , and if the text is uniform, then $|S|\subseteq \mathcal O(\log n)$ with high probability~\cite{devroye1992note}.
However, as noted in \cite[Section 5.1]{abouelhoda2004replacing}, the stack size in practice is much smaller.

\begin{proposition}\label{prop}
  During the sequential execution of $\Begin(i)$ for $i =  2,3, \dots, n$, the size of the monotone stack is bounded by $h+1$ in the worst case, where $h$ denotes the height of the suffix tree constructed over the reverse of the input text, that is the maximum number of branching nodes on any root--to--leaf path. 
\end{proposition}

\begin{proof}
Let $Trie$ denote the suffix tree constructed over $T^{rev}[1..n]$.
For each $i$, let $A_i$ denote the list of all index entries, excluding $1$, of the triples stored in the monotone stack $S$ (in increasing order) upon completion of the execution of $\Begin(i)$.
Thus, $|A_i| = |S|_i - 1$.

Consider any $j', j'' \in A_i$ with $j' < j''$, and any $s \in (j'..j'']$.
Following the algorithm, we have $\LCP[j'] < \LCP[s]$; otherwise, the triple indexed by $j'$ would have been popped from the stack.
Consequently, $T^{rev}[\SA[j']..\SA[j']+\LCP[j']-1]$ is a prefix of $T^{rev}[\SA[j'']..\SA[j'']+\LCP[j'']-1]$.

For each $j \in A_i$, let $u_j$ denote the node in $Trie$ that is the lowest common ancestor of the $\SA[j-1]$-leaf and the $\SA[j]$-leaf.
Observe that $u_j$ is the locus of the string $T^{rev}[\SA[j]..\SA[j]+\LCP[j]-1]$.
Hence, $u_{j'}$ is an ancestor of $u_{j''}$, and the nodes $\{u_j \mid j \in A_i\}$ all lie on a single root-to-leaf path in $Trie$.

It follows that $|A_i| \le h$.
Since $|A_i| = |S|_i - 1$, we conclude that $|S|_i \le h + 1$.
\end{proof}






\section{A Smallest Suffixient Set: A New Characterization and Proof}
\label{sect-characterize}
In this section, we characterize a minimum-size suffixient set based the concepts of $\Weight(\cdot), \Begin(\cdot),$ and $\End(\cdot)$ and prove its correctness.


\begin{definition}[$\Cand$]\label{def-cand}
Let $1 \le a < a' \le n$ and $c \in \Sigma$.
Define $\Cand_c(a,a')$ as the smallest index $p \in [a..a')$ with $\BWT[p]=c$ and
$\Weight(p) = \max\{\Weight(x) \mid a \le x < a',\ \BWT[x] = c,\ \Weight(x) \ge 0\}$;
if no such index exists, set $\Cand_c(a,a') := -1$.
\end{definition}



By Definition~\ref{def-cand}, for any index $x$ with $\BWT[x] \neq \BWT[x-1]$ and any $c \in \{\BWT[x], \BWT[x-1]\}$, it follows that $\Cand_c(\Begin(x), \End(x)) \neq -1$.  
Indeed, both $x$ and $x-1$ lie within the interval $[\Begin(x), \End(x))$, and $\Weight(x-1) \ge 0$ as well as $\Weight(x) \ge 0$.  
Observation~\ref{obs-weight-lcp} will be used throughout the remainder of the paper.

\begin{observation}\label{obs-weight-lcp}
For any index $1 < x \le n$ such that $\BWT[x] \neq \BWT[x-1]$, and any $c \in \{\BWT[x], \BWT[x-1]\}$, we have $\Weight(\ell_{x,c}) \ge \LCP[x]$, where $\ell_{x,c} = \Cand_c(\Begin(x), \End(x))$.
\end{observation}

\begin{proof}
Since $\BWT[x] \neq \BWT[x-1]$ and $c \in \{\BWT[x], \BWT[x-1]\}$, Definition~\ref{def-cand} implies that $\ell_{x,c} = \Cand_c(\Begin(x), \End(x)) > -1$ and that $\Weight(\ell_{x,c}) = \max\{\Weight(j) \mid \Begin(x) \le j < \End(x),\ \BWT[j] = c\}$.

Let $\ell \in \{x, x-1\}$ be the index with $\BWT[\ell] = c$. By the definition of $\Weight(\cdot)$, it follows that $\Weight(\ell) \ge \LCP[x]$. Moreover, since $\ell \in [\Begin(x), \End(x))$ by the definitions of $\Begin(\cdot)$ and $\End(\cdot)$, and $\BWT[\ell] = c$, the maximality of $\Weight(\ell_{x,c})$ implies $\Weight(\ell_{x,c}) \ge \Weight(\ell) \ge \LCP[x]$.
\end{proof}




Our goal is to show that $\FullCand$, defined in Definition~\ref{def-ell-x-c}, is a minimum-size suffixient set, thereby providing a new characterization of such sets in terms of $\Weight(\cdot)$, $\Begin(\cdot)$, and $\End(\cdot)$. To this end, one possible approach is to establish a bijection between $\FullCand$ and the set of all \emph{super-maximal extensions}, as defined in \cite[Definition~32]{cenzato2025}. This approach implies, as demonstrated in \cite[Section~5]{cenzato2025}, that $\FullCand$ is a minimum-size suffixient set. Instead, we present a new proof that does not rely on the notion of super-maximal extensions.


\begin{definition}[$\FullCand$]\label{def-ell-x-c}
Define $\FullCand = \{n-\SA[p]+1 \mid p=\Cand_c(\Begin(x), \End(x)), 1 < x \le n, \BWT[x]\ne \BWT[x-1], c \in \{\BWT[x], \BWT[x-1]\}\}$.
\end{definition}


We first prove that $\FullCand$ is suffixient for the text $T[1..n]$, applying Definition~\ref{def-suffixient-set}.

\begin{lemma}\label{lem-suffixient}
The set $\FullCand$ is suffixient for the text $T[1..n]$.
\end{lemma}

\begin{proof}
By Definition~\ref{def-suffixient-set}, it suffices to show the following:
for any right-maximal substring $str$ of $T[1..n]$ and for any one-character
right-maximal extension $str \cdot c$ of $str$, there exists a position
$x_{\text{text}} \in \FullCand$ such that $str \cdot c$ is a suffix of
$T[1..x_{\text{text}}]$.

Since $str$ is right-maximal in $T[1..n]$, there exists another character $a \neq c$
such that $str \cdot a$ also occurs as a substring of $T[1..n]$.
Consequently, there exists at least one interval $[\ell, z]$ such that both characters $a$ and $c$ occur in $\BWT[\ell..z]$, and the prefix $T[1..n-\SA[i]]$
has $str$ as a suffix for every index $i \in [\ell, z]$.

Let $[\ell, z]$ be the smallest such interval.
Without loss of generality, assume that $\BWT[\ell] \neq c$ and $\BWT[z] = c$;
the case $\BWT[\ell] = c$ and $\BWT[z] \neq c$ is symmetric.
By the minimality of the interval $[\ell, z]$, we have
$\BWT[z] \neq \BWT[z-1]$.
Let $q := \Cand_c(\Begin(z), \End(z))$.
Since $\BWT[z] = c = \BWT[q]$ and $\BWT[z] \neq \BWT[z-1]$, it follows that
$\Weight(q) \ge \Weight(z) \ge \LCP[z] \ge |str|$; moreover, as $T[n-\SA[q]+1]=\BWT[q]=c$, we have $str \cdot c$ is a suffix of $T[1..n-\SA[q]+1]$.
Setting $x_{\text{text}} := n-\SA[q]+1$, we conclude that
$x_{\text{text}} \in \FullCand$, completing the proof.
\end{proof}

\begin{lemma}\label{lem-suffix-impossible}
Define $\ell_{x,c} := \Cand_c(\Begin(x), \End(x))$ and $\ell_{x',c'} := \Cand_{c'}(\Begin(x'), \End(x'))$ for any two start-run boundaries $x$ and $x'$, and for any $c \in \{\BWT[x],\BWT[x-1]\}$ and $c' \in \{\BWT[x'], \BWT[x'-1]\}$, respectively.
If $\ell_{x,c} \neq \ell_{x',c'}$ and
$\Weight(\ell_{x,c}) \le \Weight(\ell_{x',c'})$, then
$T[n-\SA[\ell_{x,c}] + 1 - \Weight(\ell_{x,c})..\, n-\SA[\ell_{x,c}] + 1]$
cannot be a suffix of
$T[n-\SA[\ell_{x',c'}] + 1 - \Weight(\ell_{x',c'})..\, n-\SA[\ell_{x',c'}] + 1]$.
\end{lemma}

\begin{proof}
By Definition~\ref{def-cand}, we have $\ell_{x,c} > -1$,
$\Weight(\ell_{x,c}) > -1$, $\ell_{x',c'} > -1$, and
$\Weight(\ell_{x',c'}) > -1$.

Any index $\ell$ with $\Weight(\ell) \ge 0$ induces the right-maximal substring
$T[n-\SA[\ell] + 1 - \Weight(\ell),\, n-\SA[\ell]]$, which cannot contain the
symbol $\$$.
Recall that the suffix array $\SA$ is constructed over $T^{rev}[1..n]$.
Therefore, the reverse of a right-maximal substring induces an interval in the
suffix array, called its \emph{SA-interval} \cite[Section~2.2.2]{li2010fast}.
Formally, for any index $\ell$ with $\Weight(\ell) \ge 0$, let
$[s(\ell), t(\ell)]$ denote the SA-interval of the string
$T[n-\SA[\ell] + 1 - \Weight(\ell)..\, n-\SA[\ell]]$, where
$s(\ell)$ is the minimum index $i \in [1..n]$ such that the string is a suffix of
$T[1..n-\SA[i]]$, and $t(\ell)$ is defined analogously as the maximum such index.

\begin{claim}\label{claim-s-t-e-b}
We have $[s(\ell_{x,c}), t(\ell_{x,c})] \subseteq [\Begin(x), \End(x))$ and
$[s(\ell_{x',c'}), t(\ell_{x',c'})] \subseteq [\Begin(x'), \End(x'))$.
\end{claim}

\begin{claimproof}
It suffices to show that $s(\ell_{x,c}) \ge \Begin(x)$ and
$t(\ell_{x,c}) < \End(x)$.
By the definition of SA-intervals,
$s(\ell_{x,c}) = \max\{ i \in [1..\ell_{x,c}) \mid \LCP[i] < \Weight(\ell_{x,c}) \}$
and
$t(\ell_{x,c}) = \min\{ i \in [\ell_{x,c}..n] \mid \LCP[i] < \Weight(\ell_{x,c}) \} - 1$.

Assume for contradiction that $s(\ell_{x,c}) < \Begin(x)$.
Since $\Begin(x) \le \ell_{x,c}$, this implies
$\LCP[\Begin(x)] \ge \Weight(\ell_{x,c})$; otherwise,
$s(\ell_{x,c}) \ge \Begin(x)$ by definition.
By the definition of $\Begin(\cdot)$, $\LCP[\Begin(x)] < \LCP[x]$, yielding
$\LCP[x] > \Weight(\ell_{x,c})$, which contradicts
Observation~\ref{obs-weight-lcp}.

A symmetric argument shows that $t(\ell_{x,c}) < \End(x)$.
If $\End(x) \le t(\ell_{x,c})$, then since $\ell_{x,c} < \End(x)$ we obtain
$\LCP[\End(x)] \ge \Weight(\ell_{x,c})$, which together with
$\LCP[x] > \LCP[\End(x)]$ again contradicts Observation~\ref{obs-weight-lcp}.
This proves the claim.
\end{claimproof}

We now prove the lemma by contradiction.
Assume that
$T[n-\SA[\ell_{x,c}] + 1 - \Weight(\ell_{x,c})..\, n-\SA[\ell_{x,c}] + 1]$
is a suffix of
$T[n-\SA[\ell_{x',c'}] + 1 - \Weight(\ell_{x',c'})..\, n-\SA[\ell_{x',c'}] + 1]$,
which implies $c = c'$.
Consequently,
$T[n-\SA[\ell_{x,c}] + 1 - \Weight(\ell_{x,c})..\, n-\SA[\ell_{x,c}]]$
is a suffix of
$T[n-\SA[\ell_{x',c'}] + 1 - \Weight(\ell_{x',c'})..\, n-\SA[\ell_{x',c'}]]$,
and hence $\ell_{x',c'} \in [s(\ell_{x,c}), t(\ell_{x,c})]$.

If $\Weight(\ell_{x,c}) < \Weight(\ell_{x',c'})$, then by
Claim~\ref{claim-s-t-e-b} we have
$\ell_{x',c'} \in [\Begin(x), \End(x))$, which implies
$\Weight(\Cand_c(\Begin(x), \End(x))) > \Weight(\ell_{x,c})$, contradicting the definition of $\ell_{x,c}$.
If instead $\Weight(\ell_{x,c}) = \Weight(\ell_{x',c'})$, then the two substrings are equal, so $s(\ell_{x,c}) = s(\ell_{x',c'})$ and
$t(\ell_{x,c}) = t(\ell_{x',c'})$.
Assuming without loss of generality that $\ell_{x,c} < \ell_{x',c'}$, we obtain
$\Begin(x') \le s(\ell_{x',c'}) \le \ell_{x,c} < \ell_{x',c'} \le t(\ell_{x',c'}) <
\End(x')$, which contradicts the choice of $\ell_{x',c'}$ as
$\Cand_{c'}(\Begin(x'), \End(x'))$.
In both cases, we reach a contradiction, completing the proof.
\end{proof}

Finally, we show that $\FullCand$ has minimum possible size, applying the \emph{pigeonhole principle}.

\begin{lemma}\label{lem-min-size}
    The list $\FullCand$ is a minimum-size suffixient set. 
\end{lemma}

\begin{proof}
Let $F$ be an arbitrary suffixient set. Our goal is to show that $|F|\ge |\FullCand|$.

Suppose, for the sake of contradiction, that $|F|<|\FullCand|$.
Recall that each $x_{\text{text}}\in \FullCand$ corresponds to a right-maximal extension
$T[x_{\text{text}}-\Weight(\SA^{-1}[n-x_{\text{text}}+1])..\, x_{\text{text}}]$.
By the definition of suffixient sets, the string
$T[x_{\text{text}}-\Weight(\SA^{-1}[n-x_{\text{text}}+1])..\, x_{\text{text}}]$
must be a suffix of $T[1..f_{\text{text}}]$ for some $f_{\text{text}}\in F$.

Since $|F|<|\FullCand|$, the pigeonhole principle implies that there exist two distinct indices
$x_{\text{text}},y_{\text{text}}\in \FullCand$ such that both
$T[x_{\text{text}}-\Weight(\SA^{-1}[n-x_{\text{text}}+1])..\, x_{\text{text}}]$ and
$T[y_{\text{text}}-\Weight(\SA^{-1}[n-y_{\text{text}}+1])..\, y_{\text{text}}]$
are suffixes of the same prefix $T[1..f_{\text{text}}]$ for some $f_{\text{text}}\in F$.
Consequently, one of these two strings must be a suffix of the other.

However, by Lemma~\ref{lem-suffix-impossible}, no such pair $x_{\text{text}}$ and $y_{\text{text}}$ can exist.
This contradiction shows that our assumption was false, and therefore $|F|\ge |\FullCand|$.
\end{proof}

Lemmas~\ref{lem-min-size} and \ref{lem-suffixient} together imply that $\FullCand$ is a minimum-size suffixient set.

\section{Computing $sA$: A One-Pass Algorithm with One-Step Look-ahead}
\label{sect-computing-sa}

In this section, we present a new one-pass algorithm that computes suffixient arrays.
The algorithm iterates over the index \(i\) from \(1\) to \(n\) and makes decisions upon reading the triple $(\BWT[i], \SA[i], \LCP[i])$.
During the iteration at index \(i\), the algorithm may additionally access the input of the next iteration, namely the triple
$(\BWT[i+1], \SA[i+1], \LCP[i+1])$,
but it requires no information about inputs beyond \(i+1\).
For this reason, we refer to it as a \emph{one-pass algorithm with one-step look-ahead}.


An overview of this section is as follows.
Section~\ref{sect-ds} introduces the data structures and the three invariants maintained by the algorithm.
Section~\ref{sect-alg} then presents a detailed description of the algorithm together with the underlying intuitions.
Since the algorithm invokes the operation $\Cand$ at each iteration, it may initially appear to require quadratic time.
Before addressing the complexity analysis, we establish the correctness of the algorithm in Section~\ref{sect-proof}.
Next, Section~\ref{sect-adjust-alg} shows that a \emph{validity test} for $\Cand$, such as checking whether $i = \Cand(\cdot,\cdot)$, can be implemented using a constant number of amortized $\mathcal O(1)$-time operations, thereby justifying that the algorithm runs in a one-pass fashion.
Finally, Section~\ref{sect-complexity} presents the overall complexity analysis.

\subsection{The Data Structures}
\label{sect-ds}

During the execution of the algorithm, we maintain the following data structures:
\begin{itemize}
    \item A monotone stack used to compute $\Begin(i)$ during the $i$-th iteration of the algorithm;
    \item A doubly-linked list, $\RowList$, containing at most $\sigma$ triples
        drawn from $[1..n] \times \Sigma \times [0..n]$;
    \item An array $\Map[1..\sigma]$ such that, for each $c \in \Sigma$, the entry $\Map[c]$ stores a pointer to the (unique) triple $(\cdot, c, \cdot)$ in $\RowList$, if such a triple exists, and stores \texttt{null} otherwise;
    \item For every $c \in \Sigma$, a linked list $\Result_c$ containing positions drawn from $[1..n]$.
\end{itemize}
For convenience, we define $\Result$ as the list obtained by concatenating all 
$\Result_1, \dots, \Result_\sigma$ in this order.

For every triple $(p_\text{text}, \cdot, \cdot)$ stored in $\RowList$, we will refer to
$p_\text{text}$ as a \emph{candidate position}, and to the corresponding index $\SA^{-1}[n - p_{\text{text}} + 1]$ as a \emph{candidate index}.

The following invariants describe the content of the above data-structures
during the $i$-th iteration of the algorithm.
Among these, \textbf{Invariant~2} characterizes the candidate indices stored in $\RowList$.
\begin{itemize}
    \item \textbf{Invariant 1:} For each \(c \in \Sigma\), $\Result_c$ is the sub-list of $\FullCand$ containing exactly those positions \(j_{\text{text}}\) such that \(\SA^{-1}[n-j_{\text{text}}+1] < i\) and \(T[j_{\text{text}}] = c\).
    \item \textbf{Invariant~2:} The list $\RowList$ contains exactly those triples $(n - \SA[j] + 1,\, \BWT[j],\, \Weight(j))$ where the index $j$ satisfies either of the following conditions:
\begin{itemize}
    \item $\Weight(j) = \LCP[j]$, $j \le i-1 < \End(j)$, and
    $\Cand_{\BWT[j]}(\Begin(j), i) = j$; or
    \item $\Weight(j) > -1$, $\Weight(j) \neq \LCP[j]$, $j \le i-1 < \End(j+1)$, and
    $\Cand_{\BWT[j]}(\Begin(j+1), i) = j$.
\end{itemize}

    \item \textbf{Invariant~3:} The triples in $\RowList$ are sorted in non-increasing order by weight.
\end{itemize}


\subsection{The Algorithm}
\label{sect-alg}

In this section we give an high-level description of the algorithm. As already
mentioned above, the algorithm iterates over indexes $i=1, 2, \dots, n$ once.
In Step 0 we describe the special case of $i=1$. Then, Step 1 and Step 2 are
executed in every successive iteration $i = 2, 3 \dots, n$. Finally, a last step
will return the list $\FullCand$ using the information computed in the preceding
iterations.

For both Step 1 and 2 we give some intuitive explanation on why they are correct.
The full proof is given in Section~\ref{sect-proof}. The pseudocode of the algorithm is deferred to the Appendix~\ref{app-pseudo-sa}.


\subparagraph*{Step 0.}
In the first iteration, if $\BWT[2] \ne \BWT[1]$ we append the triple 
$(n - \SA[1] + 1, \BWT[1], \Weight(1))$ to the front of $\RowList$, and we set 
$\Map[\BWT[1]]$ to point to such triple; otherwise, no triple is inserted in
$\RowList$. Every list $\Result_c$ remains empty.

We now consider iteration $i$ for $i\in (1, n]$.
The operations are divided into two steps.

\subparagraph*{Step 1.}
Remove from $\RowList$ every triple $(p_{\text{text}}, c, w)$ such that
$w > \LCP[i]$.
Set $\Map[c]$ to \texttt{null} and append the position
$p_{\text{text}}$ to the tail of the list $\Result_{c}$.



\subparagraph*{Intuition behind Step~1.}
We provide some intuition for why the position $p_{\text{text}}$ appended to
$\Result_{char}$ always belongs to $\FullCand$, a property required to maintain
\textbf{Invariant~1}.

Let $j = \SA^{-1}[n - p_{\text{text}} + 1]$.
By Lemma~\ref{lem-LCP-candidate}, depending on whether $\LCP[j] = w$,
we have either $j = \Cand_{c}(\Begin(j), i)$ or
$j = \Cand_{c}(\Begin(j+1), i)$.
Moreover, Lemma~\ref{lem-end-j-j+1} shows that when $\LCP[i] < w$,
the corresponding interval endpoint satisfies
$\End(j) = i$ in the former case and $\End(j+1) = i$ in the latter.
Thus, intuitively, $j$ is always selected as a candidate index over one of the intervals $[\Begin(j), \End(j))$ or $[\Begin(j+1), \End(j+1))$.

Furthermore, Lemma~\ref{lem-result_c-to-L_c} guarantees that in the first case
$\BWT[j] \neq \BWT[j-1]$, while in the second case $\BWT[j+1] \neq \BWT[j]$.
In both situations, $j$ satisfies the defining conditions of $\FullCand$.
Therefore, the position $p_{\text{text}}$ appended to $\Result_{c}$ indeed
belongs to $\FullCand$.

Lemma~\ref{observe-2} states that the triples in $\RowList$ are sorted in non-increasing order of their weight values.
As a result, every triple satisfying $\LCP[i] < w$ can be
identified and removed in $\mathcal O(1)$ time.

\subparagraph*{Step~2.}
Compute the value $\Weight(i)$.
If $\Weight(i) = -1$, proceed immediately to the next iteration.
Otherwise, compute the index $\Begin(i)$ and initialize a variable
$\CurrB$  to $\Begin(i)$.
If $\Weight(i) \neq \LCP[i]$, update $\CurrB$ to $\Begin(i+1)$.

Assume that $c = \BWT[i]$.
We then check whether $i = \Cand_c(\CurrB, i+1)$.
If this condition does not hold, proceed immediately to the next iteration.
If $\Map[c] \neq \texttt{null}$, remove from $\RowList$ the triple pointed to by $\Map[c]$.
Then prepend the triple $(n-\SA[i]+1, c, \Weight(i))$ to $\RowList$ and update
$\Map[c]$ to point to the new head.

\subparagraph*{Intuition behind Step~2.}
We now give an informal explanation of why both \textbf{Invariant~2} and
\textbf{Invariant~3} are maintained after this step.
If either $\Weight(i) = -1$ or $i \neq \Cand_c(\CurrB, i+1)$, the algorithm
immediately proceeds to the next iteration.
In this situation, no triple is added to $\RowList$, and both invariants
are trivially maintained.
Otherwise, the triple
$(n - \SA[i] + 1, \BWT[i], \Weight(i))$ is added to the front of $\RowList$.
By Lemma~\ref{observe-2}, this insertion preserves the non-increasing order
of weights in $\RowList$, ensuring that \textbf{Invariant~3} holds.

It remains to argue that the newly added triple satisfies one of the
conditions in \textbf{Invariant~2}.
Two cases arise, depending on whether $\LCP[i] = \Weight(i)$.
If $\LCP[i] = \Weight(i)$, then $\CurrB = \Begin(i)$ and
$i = \Cand_c(\Begin(i), i+1)$; otherwise, $\LCP[i] \neq \Weight(i)$,
$\CurrB = \Begin(i+1)$, and $i = \Cand_c(\Begin(i+1), i+1)$.
Setting $j := i$, we have $\Weight(j) > -1$, and either
$j \le i+1-1 < \End(j)$ with $\Cand_c(\Begin(j), i+1) = j$, or
$j \le i+1-1 < \End(j+1)$ with $\Cand_c(\Begin(j+1), i+1) = j$.
In both cases, the conditions of \textbf{Invariant~2} are satisfied.
In Section~\ref{sect-adjust-alg}, we show how the test
$i = \Cand_c(\CurrB, i+1)$ can be implemented in amortized constant time.

\subparagraph*{Final step after the $n$ iterations.}  
If $\RowList$ is nonempty, append all positions $p_{\text{text}}$ from the remaining tuples $(p_{\text{text}}, c, w)$ to their respective $\Result_c$ lists.  
Then all the $\Result_c$ lists are concatenated to obtain $\Result$, which is then
returned as the suffixient array.


In the next section we prove that $\Result = \FullCand$.
Hence, by Lemma~\ref{lem-min-size}, the output list is suffixient and of minimum size.
Note that the prefix list
$\{\,T[1..\SA[1]], \dots, T[1..\SA[n]]\,\}$
is sorted in co-lexicographic order.
We further show that this ordering ensures that the indices in the output list are correctly sorted, in accordance with the definition of a suffixient array.

\subsection{Correctness of the Algorithm}
\label{sect-proof}

In this section, we prove the correctness of the algorithm.
The proof is divided into three parts.
In Section~\ref{sect-order} we prove that the triples in $\RowList$ are always sorted non-increasingly by their weight values.
Then, in Section~\ref{sect-proof-suff}, we show that for any character $c\in\Sigma$ and any position $\ell_{\text{text}} \in \Result_c$, it holds that $\ell_{\text{text}} \in \FullCand$.
Finally, in Section~\ref{sect-proof-ness}, we prove the converse: every $\ell_{\text{text}} \in \FullCand$ is added to $\Result_c$ for some $c\in\Sigma$; we also show that the order of indices in the final list $\Result$ is consistent with the order induced by the suffixient array.

\subsubsection{The Ordering on the Triples in $\RowList$}
\label{sect-order}

\begin{lemma}\label{lem-consistency}
Let $c = \BWT[i]$. If $i = \Cand_c(\CurrB, i+1)$ at iteration $i$ and $\Weight(i) \le \LCP[i]$, then one of the following holds: either $\BWT[j] = c$ for every $j \in [\CurrB, i)$, or $\BWT[j] \ne c$ for every $j \in [\CurrB, i)$.
\end{lemma}

\begin{proof}
Suppose, for the sake of contradiction, that there exists an index $j' \in [\CurrB, i-1)$ such that either $\BWT[j'] = c$ and $\BWT[j' + 1] \ne c$, or $\BWT[j'] \ne c$ and $\BWT[j' + 1] = c$. Let $j'' \in \{j', j' + 1\}$ be such that $\BWT[j''] = c$.

We claim that $\Weight(j'') \ge \LCP[j' + 1]$. Indeed, if $\BWT[j'] = c$ and $\BWT[j' + 1] \ne c$, then by the definition of $\Weight(\cdot)$ we have $\Weight(j'') = \Weight(j') \ge \LCP[j' + 1]$. Otherwise, if $\BWT[j'] \ne c$ and $\BWT[j' + 1] = c$, then $\Weight(j'') = \Weight(j' + 1) \ge \LCP[j' + 1]$.

Define $\ell := i$ if $\Weight(i) = \LCP[i]$, and otherwise $\ell := i + 1$. By construction, $\Begin(\ell) = \CurrB$ and $\Weight(i) = \LCP[\ell]$. Since $j' + 1 \in (\CurrB, i)$ and $(\CurrB, i) \subseteq (\Begin(\ell), \ell)$, the definition of $\Begin(\cdot)$ implies $\LCP[j' + 1] \ge \LCP[\ell]$. Consequently, $\Weight(j'') \ge \LCP[j' + 1] \ge \LCP[\ell] = \Weight(i)$.

Finally, since $\BWT[j''] = c$, $\CurrB = \Begin(\ell) \le j'' < i$, and $\Weight(j'') \ge \Weight(i)$, this contradicts $i = \Cand_c(\CurrB, i + 1)$. Hence, the assumption is false, and the lemma follows.
\end{proof}

Note that in Lemma~\ref{lem-consistency}, if $\Weight(i) > \LCP[i]$, then
$\Weight(i) = \LCP[i+1] > \LCP[i]$ and $\CurrB = \Begin(i+1) = i$; hence,
the interval $[\CurrB, i)$ is empty.

\begin{lemma}\label{observe-2}
The entries in $\RowList$ are sorted in non-increasing order by weight.
\end{lemma}

\begin{proof}
We prove the claim by induction on the iteration index.

In the first iteration, if $\BWT[1] = \BWT[2]$, no triple is added to $\RowList$, so it remains empty. Otherwise, $\RowList$ contains exactly one triple. In both cases, the claim holds trivially.

Assume as induction hypothesis that after each of the first $(i-1)$ iterations (for some $i > 1$), the triples in $\RowList$ are sorted in non-increasing order of their $\Weight$ values.

Now consider the $i$-th iteration.  
Any triple removed from $\RowList$ in Step~1 is always removed from the front and has weight strictly greater than $\LCP[i]$. Therefore, after Step~1, the relative order of the remaining triples is preserved. Let $(p_{\text{text}}, c, w)$ denote the triple currently at the front of $\RowList$ (if any); then we have $\LCP[i] \ge w$.

If no triple is added to $\RowList$ in Step~2, the claim follows immediately from the induction hypothesis. Otherwise, we must show that the newly added triple $(n-\SA[i] + 1, \BWT[i], \Weight(i))$ satisfies $\Weight(i) \ge w$.
Since $\LCP[i] \ge w$, if $\Weight(i) \ge \LCP[i]$, then $\Weight(i) \ge w$ holds immediately. 

It remains to consider the case $\Weight(i) < \LCP[i]$.
Since the triple is added at iteration $i$, we have $i = \Cand_c(\CurrB, i+1)$ for $c = \BWT[i]$. Because both $i = \Cand_c(\CurrB, i+1)$ and $\Weight(i) < \LCP[i]$, Lemma~\ref{lem-consistency} applies, so either $\BWT[j] = c$ for every $j \in [\CurrB, i)$ or $\BWT[j] \ne c$ for every such $j$.
Since $i > 1$ and $\Weight(i) < \LCP[i]$, by definition of $\Weight(\cdot)$ we must have $\BWT[i-1] = \BWT[i]$, which implies $\BWT[j] = c$ for every $j \in [\CurrB, i)$. Consequently, $\Weight(j) = -1$ for every $j \in (\CurrB, i)$, and either $\Weight(\CurrB) = \LCP[\CurrB]$ or $\Weight(\CurrB) = -1$.
The former case implies $\SA^{-1}[n-p_{\text{text}}+1] \le \CurrB$. Moreover, since $\SA^{-1}[n-p_{\text{text}}+1] \le \CurrB$ and $\Weight(\CurrB) \le \LCP[\CurrB]$, it follows that $w \le \LCP[\CurrB]$.
On the other hand, since $\Weight(i) < \LCP[i]$, we have $\Weight(i) = \LCP[i+1]$ and $\CurrB = \Begin(i+1)$. Because $\LCP[i+1] > \LCP[\CurrB]$ and $w \le \LCP[\CurrB]$, it follows that $\Weight(i) > w$.

Therefore, whether $\Weight(i) \ge \LCP[i]$ or not, we have $\Weight(i) \ge w$ (in fact, strictly greater in the latter case), so inserting the new triple preserves the non-increasing order in $\RowList$. This completes the induction and the proof.
\end{proof}

\subsubsection{Necessity: $\Result \subseteq \FullCand$}
\label{sect-proof-suff}

Throughout this section, let $\RowList_i$ denote the state of $\RowList$ at the start of the $i$-th iteration of the algorithm. Let $\POS_i = \{\SA^{-1}[n - p_{\text{text}} + 1] \mid (p_\text{text}, \cdot, \cdot) \in \RowList_i\}$. For each $j \in \POS_i$, define $x_j \in \{\,j, j+1\,\}$ by $x_j = j$ if $\Weight(j) = \LCP[j]$, and $x_j = j+1$ otherwise.
Observation~\ref{obser-x_i-curr_b} details the properties of $x_j$.

\begin{observation}\label{obser-x_i-curr_b}
For each $j \in \POS_i$, the following hold: (a) either $\BWT[j] \neq \BWT[j-1]$ or $\BWT[j] \neq \BWT[j+1]$; (b) $\Weight(j) = \LCP[x_j]$; and (c) $\Begin(x_j) = \CurrB$ at iteration $j$. 
\end{observation}

\begin{proof}
Since $j \in \POS_i$, the algorithm guarantees $\Weight(j) \ge 0$. By definition of $\Weight(\cdot)$, this implies $\Weight(j) \in \{\,\LCP[j], \LCP[j+1]\,\}$ and that either $\BWT[j] \neq \BWT[j-1]$ or $\BWT[j] \neq \BWT[j+1]$, proving (a), and $\Weight(j) = \LCP[x_j]$, proving (b).

If $\Weight(j) \neq \LCP[j]$, then $x_j = j+1$ and $\CurrB = \Begin(j+1)$, so $\CurrB = \Begin(x_j)$ at iteration $j$. Otherwise, $x_j = j$ and $\CurrB = \Begin(j)$, so again $\CurrB = \Begin(x_j)$. This proves (c).
\end{proof}

%

\begin{lemma}\label{lem-LCP-candidate}
    For every $j \in \POS_i$, if $\Weight(j)\ne \LCP[j]$, then $\Cand_c(\Begin(j+1), i)=j$, where $c=\BWT[j]$; otherwise, $\Cand_c(\Begin(j), i)=j$.
\end{lemma}

\begin{proof}
    Consider the $j$-th iteration of the algorithm.
    Since $j\in \POS_i$ and $j<i$, the tuple $(n-\SA[j]+1, \BWT[j], \Weight(j))$ is added to $\RowList$ in the $j$-th iteration.
    As shown in the algorithm, we know that at iteration $j$, the condition $j=\Cand_c(\CurrB, j+1)$ holds; moreover, by Observation~\ref{obser-x_i-curr_b}, we have $\CurrB=\Begin(x_j)$, so, $j=\Cand_c(\Begin(x_j), j+1)$.

    Since $j\in \POS_i$, we know that for any index $j<j'<i$, $\LCP[j']$ cannot be smaller than $\Weight(j)$ (otherwise, the tuple with $weight=\Weight(j)$ would be removed from $\RowList$ before iteration $i$); moreover, if $\BWT[j']=c$, then $\Weight(j')\le \Weight(j)$.
    Therefore, we have $j=\Cand_c(\Begin(x_j), i)$.
    If $\Weight(j)\ne \LCP[j]$, then $x_j=j+1$, so $\Cand_c(\Begin(j+1), i)=j$; otherwise, $x_j=j$, and $\Cand_c(\Begin(j), i)=j$, completing the proof.
\end{proof}

\begin{lemma}\label{lem-end-j-j+1} 
    For every $j \in \POS_i$ with $\LCP[i]<\Weight(j)$, if $\Weight(j)\ne \LCP[j]$, then $\End(j+1)=i$; otherwise, $\End(j)=i$. 
\end{lemma}

\begin{proof}
Since $j\in \POS_i$, we know that $\LCP[j']\ge \Weight(j)$ for any index $j<j'<i$ (otherwise, $j\notin \POS_i$, a contradiction).
Recall that $\Weight(j)=\LCP[x_j]$, where $x_j\in \{j, j+1\}$, by Observation~\ref{obser-x_i-curr_b}.
Since $\Weight(j)>\LCP[i]$ and $\Weight(j)=\LCP[x_j]$, we have $\LCP[x_j]>\LCP[i]$, so $i>x_j$ is the smallest index with $\LCP[i]<\LCP[x_j]$, so, we have $\End(x_j)=i$.
If $\Weight(j)\ne \LCP[j]$, then $\Weight(j)=\LCP[j+1]$ and $x_j=j+1$, so $\End(j+1)=i$;
otherwise, $\Weight(j)=\LCP[j]$ and $x_j=j$, so $\End(j)=i$.
\end{proof}

Lemma~\ref{lem-result_c-to-L_c} and Definition~\ref{def-ell-x-c} imply $\ell_{\text{text}} \in \FullCand$ for any $\ell_{\text{text}} \in \Result_c$ and any $c \in [1, \sigma]$. 

\begin{lemma}\label{lem-result_c-to-L_c}
Let $(p_{\text{text}}, c, weight)$ be any triple removed from $\RowList_i$ in the first step of the $i$-th iteration of the algorithm, and let $j = \SA^{-1}[n - p_{\text{text}} + 1]$.
Either $\Cand_c(\Begin(j+1), \End(j+1)) = j$ and $\BWT[j+1]\ne \BWT[j]$, or $\Cand_c(\Begin(j), \End(j)) = j$ and $\BWT[j]\ne \BWT[j-1]$.
\end{lemma}

\begin{proof}
Since $p_{\text{text}}$ is removed from $\RowList_i$, we have $\Weight(j) \ge 0$.
Thus, either $\Weight(j) = \LCP[j]$ or $\Weight(j) = \LCP[j+1] \neq \LCP[j]$.

If $\Weight(j) \neq \LCP[j]$, then this implies $\BWT[j+1] \neq \BWT[j]$, by the definition of $\Weight(\cdot)$.
By Lemmas~\ref{lem-LCP-candidate} and~\ref{lem-end-j-j+1}, we obtain $\Cand_c(\Begin(j+1), \End(j+1)) = j$; the statement holds. 

Otherwise, $\Weight(j) = \LCP[j]$, this implies that either $\BWT[j] \neq \BWT[j-1]$, or $\BWT[j-1] = \BWT[j] \neq \BWT[j+1]$ and $\LCP[j+1] = \LCP[j]$.
In both cases, Lemmas~\ref{lem-LCP-candidate} and~\ref{lem-end-j-j+1} imply $\Cand_c(\Begin(j), \End(j)) = j$.
Since $\Cand_c(\Begin(j), \End(j)) = j$, if $\BWT[j] \neq \BWT[j-1]$, then the statement holds immediately.
Otherwise, since $\LCP[j+1] = \LCP[j]$, the definitions of $\Begin(\cdot)$ and $\End(\cdot)$ imply $\Begin(j+1) = \Begin(j)$ and $\End(j+1) = \End(j)$, and hence $\Cand_c(\Begin(j+1), \End(j+1)) = j$.
So, we have $\BWT[j] \neq \BWT[j+1]$ and $\Cand_c(\Begin(j+1), \End(j+1)) = j$, completing the proof.
\end{proof}


\subsubsection{Sufficiency: $\FullCand \subseteq \Result$}
\label{sect-proof-ness}


Throughout this section, let $x$ be any index with $1 < x \le n$ such that
$\BWT[x] \neq \BWT[x-1]$.
Let $c \in \{\,\BWT[x], \BWT[x-1]\,\}$, and define
$\ell_{x,c} := \Cand_c(\Begin(x), \End(x))$, so $\ell_{x, c}\in [\Begin(x), \End(x))$.

By Definition \ref{def-ell-x-c}, it suffices to prove that $n - \SA[\ell_{x,c}] + 1 \in \Result_c$.
To this end, we first prove in Lemma~\ref{lem-add-to-rowlist} that the triple $(n-\SA[\ell_{x, c}]+1, c, \Weight(\ell_{x, c}))$ is added to $\RowList$ in the $\ell_{x,c}$-th iteration of the algorithm.
We then show in Lemma~\ref{lem-remove-from-rowlist} that the tuple $(n-\SA[\ell_{x, c}]+1, c, \Weight(\ell_{x, c}))$ is removed from $\RowList$ in the $\End(\ell_{x,c})$-th iteration and subsequently added to $\Result_c$, which completes the proof.
Finally, we present in Lemma \ref{lem-order} that the ordering of the indices in the output list is consistent with the suffixient array.

\begin{lemma}\label{lem-add-to-rowlist}
The tuple $(n-\SA[\ell_{x,c}]+1, c, \Weight(\ell_{x,c}))$ is added to the front of
$\RowList$ during the $\ell_{x,c}$-th iteration of the algorithm.
\end{lemma}

\begin{proof}
Define $\ell \in \{\ell_{x,c}, \ell_{x,c}+1\}$ as follows: let
$\ell=\ell_{x,c}$ if $\Weight(\ell_{x,c})=\LCP[\ell_{x,c}]$, and
$\ell=\ell_{x,c}+1$ otherwise.
By construction, $\Weight(\ell_{x,c})=\LCP[\ell]$, and at iteration
$\ell_{x,c}$ we have $\CurrB=\Begin(\ell)$.

Since $\ell \in \{\ell_{x,c}, \ell_{x,c}+1\}$, it follows that
$\Begin(\ell)\le \ell_{x,c}$.
Moreover, by Observation~\ref{obs-weight-lcp}, $\Weight(\ell_{x,c})\ge \LCP[x]$.
Together with $\Weight(\ell_{x,c})=\LCP[\ell]$, this implies
$\LCP[\ell]\ge \LCP[x]$.
Because $\Begin(x)\le \ell_{x,c}$, $\Begin(\ell)\le \ell_{x,c}$, and
$\LCP[\ell]\ge \LCP[x]$, we obtain $\Begin(\ell)\ge \Begin(x)$.
Consequently,
$\ell_{x,c}\in [\Begin(\ell), \ell_{x,c}+1)\subseteq [\Begin(x), \End(x))$.
Therefore, $\ell_{x,c}=\Cand_c(\Begin(\ell), \ell_{x,c}+1)$.
Since $\CurrB=\Begin(\ell)$ at iteration $\ell_{x,c}$, we conclude that
$\ell_{x,c}=\Cand_c(\CurrB, \ell_{x,c}+1)$.
Hence, the tuple $(n-\SA[\ell_{x,c}]+1, c, \Weight(\ell_{x,c}))$ is added to the
front of $\RowList$ during the $\ell_{x,c}$-th iteration.
\end{proof}

Lemma \ref{lem-remove-from-rowlist} establishes that $\SA[\ell_{x, c}]+1\in \Result_c$, and thus $\FullCand\subseteq \Result$.

\begin{lemma}\label{lem-remove-from-rowlist}
The position $n-\SA[\ell_{x,c}]+1$ belongs to $\Result_c$.
\end{lemma}

\begin{proof}
By Lemma~\ref{lem-add-to-rowlist}, the tuple
$(n-\SA[\ell_{x,c}]+1, c, \Weight(\ell_{x,c}))$
is added to the front of $\RowList$ at the $\ell_{x,c}$-th iteration.
Let $j>\ell_{x,c}$ be the smallest index such that
$\LCP[j]<\Weight(\ell_{x,c})$.
By Observation~\ref{obs-weight-lcp}, we have
$\Weight(\ell_{x,c})\ge \LCP[x]$, and by the definition of $\End(x)$,
$\LCP[x]>\LCP[\End(x)]$.
Hence, $\Weight(\ell_{x,c})>\LCP[\End(x)]$, which implies
$\ell_{x,c}<j\le \End(x)$.

We first show that the tuple
$(n-\SA[\ell_{x,c}]+1, c, \Weight(\ell_{x,c}))$
remains in $\RowList$ throughout all iterations $j'$ with
$\ell_{x,c}<j'<j$.
Indeed, for any such iteration $j'$, since
$\LCP[j']\ge \Weight(\ell_{x,c})$, the tuple cannot be removed in
Step~1 of the algorithm.
In Step~2, if $\BWT[j']\neq c$, the tuple is unaffected, as only the
tuple pointed to by $\Map[\BWT[j']]$ may be updated.
If $\BWT[j']=c$, then $\Weight(j')\le \Weight(\ell_{x,c})$, because
$\ell_{x,c}=\Cand_c(\Begin(x),\End(x))$ and
$\ell_{x,c}<j'<j\le \End(x)$.
Assuming for contradiction that
$j'=\Cand_c(\CurrB, j'+1)$ at iteration $j'$, we would have
$\CurrB>\ell_{x,c}$.
Moreover, $\LCP[\CurrB]<\Weight(j')$, since either
$\Weight(j')=\LCP[j']$ and $\CurrB=\Begin(j')$, or
$\Weight(j')=\LCP[j'+1]$ and $\CurrB=\Begin(j'+1)$.
This yields
$\ell_{x,c}<\CurrB\le j'<j$ and
$\LCP[\CurrB]<\Weight(j')\le \Weight(\ell_{x,c})$,
contradicting the minimality of $j$.
Therefore, the tuple is not removed in any iteration $j'$ with
$\ell_{x,c}<j'<j$.

If $j\le n$, then the tuple is removed from
$\RowList$ in Step~1 at iteration $j$, as $\LCP[j]<\Weight(\ell_{x,c})$, and $n-\SA[\ell_{x,c}]+1$ is appended to $\Result_c$.
Otherwise, no such index $j$ exists and the tuple remains in
$\RowList$ after $n$ iterations.
Then, the tuple is removed in the final step and its position entry is added to $\Result_c$.
Hence, $n-\SA[\ell_{x,c}]+1\in \Result_c$, as claimed.
\end{proof}


Lemma~\ref{lem-order} specifies the ordering of the output list.

\begin{lemma}\label{lem-order}
Let $\Result = \{\,r_1, r_2, \dots, r_k\,\}$. Then, the prefixes $T[1..r_1]$,  $T[1..r_2]$, $\dots$, $T[1..r_k]$ are sorted in co-lexicographical order. 
\end{lemma}

\begin{proof}
Recall that $\SA[1..n]$ is the suffix array of $T^{rev}[1..n]$.
Therefore, the prefixes $T[1..n-\SA[1]], T[1..n-\SA[2]], \dots, T[1..n-\SA[n]]$ are sorted in co-lexicographical order.
This implies that for any character $c \in \Sigma$ and any positions $i, j \in \Result_c$ with $i < j$, we have
$T[1..n-\SA[i]] \prec_{\mathrm{colex}} T[1..n-\SA[j]]$, that is, $T[1..n-\SA[i]]$ is co-lexicographically smaller than $T[1..n-\SA[j]]$.
Since $T[n-\SA[i]+1] = T[n-\SA[j]+1] = c$, it follows that
$T[1..n-\SA[i]+1] \prec_{\mathrm{colex}} T[1..n-\SA[j]+1]$.

Moreover, for any characters $c, c' \in \Sigma$ with $c < c'$, and for any positions $i \in \Result_c$ and $j \in \Result_{c'}$, we have
$T[1..n-\SA[i]+1] \prec_{\mathrm{colex}} T[1..n-\SA[j]+1]$.

Combining the arguments above establishes the lemma.
\end{proof}




Lemmas~\ref{lem-result_c-to-L_c} and~\ref{lem-remove-from-rowlist} together imply that the output list $\Result$ is exactly $\FullCand$. By Lemma~\ref{lem-min-size}, $\Result$ is a minimum-size suffixient set, and Lemma~\ref{lem-order} shows that its indices are sorted consistently with the definition of suffixient array, completing the proof of the correctness. 

\subsection{The Adjusted Algorithm}
\label{sect-adjust-alg}
We have shown that the algorithm in Section~\ref{sect-alg} always computes a suffixient array correctly.
But, we have not specified how to check whether or not $i=\Cand_c(\CurrB, i+1)$ at iteration $i$, where $c=\BWT[i]$.
In this section, we specify the procedures to verify $i=\Cand_c(\CurrB, i+1)$.

In the data structure part, we construct, in addition, an array $\posprevWeight[1..\sigma]$ consisting of $\sigma$ entries, where each entry is a pair of the form $(index, weight)$ drawn from $\{1,\dots,n\} \times \{-1,0,\dots,n\}$, initially set to $(0, -1)$.
At iteration $i$, for $i>1$, the following invariant maintains: For each \(c \in \Sigma\), the field $\posprevWeight[c].index$ stores the largest index $j < i$ such that $\Weight(j) > -1$ and $\BWT[j] = c$, and $\posprevWeight[c].weight$ stores the corresponding value $\Weight(j)$. If no such index exists, the entry $\posprevWeight[c]$ is set to $(0, -1)$.


It holds that $i = \Cand_c(\CurrB, i+1)$ if and only if at least one of the following conditions is satisfied:
\textbf{C1:} $\Weight(i) > \LCP[i]$;
\textbf{C2:} $\Weight(i) > -1$ and $\posprevWeight[c].index < \CurrB$; or
\textbf{C3:} $\posprevWeight[c].weight < \Weight(i)$.
Each condition can be checked in constant time in a one-pass setting. In the remainder of this section, we prove this equivalence.

\begin{lemma}
    If $i=\Cand_c(\CurrB, i+1)$ at iteration $i$, then one of \textbf{C1}-\textbf{C3} holds.
\end{lemma}

\begin{proof}
Since $i=\Cand_c(\CurrB, i+1)>-1$, Definition~\ref{def-cand} implies that $\Weight(i)>-1$.
If $\Weight(i)>\LCP[i]$, then \textbf{C1} trivially holds.
Otherwise, by Lemma~\ref{lem-consistency}, either $\BWT[j]=c$ for every $j\in [\CurrB, i)$ or $\BWT[j]\ne c$ for every $j\in [\CurrB, i)$.
In either case, observe that $\Weight(j')=-1$ for every $j'\in (\CurrB, i)$ with $\BWT[j']=c$, so $\posprevWeight[c].index \le \CurrB$ at iteration $i$.

If $\posprevWeight[c].index<\CurrB$, then \textbf{C2} holds.
Otherwise, $\posprevWeight[c].index=\CurrB$, and $\BWT[\CurrB]=c$.
Applying Lemma~\ref{lem-consistency}, the fact $\BWT[\CurrB]=c$ implies that $\BWT[j]=c$ for every $j\in [\CurrB, i)$, so $\posprevWeight[c].weight=\Weight(\CurrB)=\LCP[\CurrB]$.

Let $\ell:=i$ if $\Weight(i)=\LCP[i]$; otherwise, $\ell:=i+1$, so $\CurrB=\Begin(\ell)$ at iteration $i$ and $\Weight(i)=\LCP[\ell]$.
As $\LCP[\ell]>\LCP[\Begin(\ell)]=\LCP[\CurrB]=\posprevWeight[c].weight$, \textbf{C3} holds.

Overall, one of three condition must hold, completing the proof.
\end{proof}

\begin{lemma}
If at iteration $i$ at least one of the conditions \textbf{C1}--\textbf{C3} holds,
then $i=\Cand_c(\CurrB, i+1)$.
\end{lemma}

\begin{proof}
We consider the three conditions separately.

\textbf{Case C1: $\Weight(i)>\LCP[i]$.}
By the definition of $\Weight(\cdot)$, this implies $\Weight(i)=\LCP[i+1]>\LCP[i]$,
and hence $\CurrB=\Begin(i+1)=i$.
Therefore, $\Cand_c(\CurrB, i+1)=\Cand_c(i, i+1)=i$. 

\textbf{Case C2: $\Weight(i)>-1$ and $\posprevWeight[c].index<\CurrB$.}
By the invariant maintained by the data structure $\posprevWeight$, this implies
that $\Weight(j)=-1$ for every $j\in[\CurrB,i)$; otherwise,
$\posprevWeight[c].index$ would be at least $\CurrB$.
Hence, $i$ is the smallest index in $[\CurrB,i]$ with nonnegative weight, and
therefore $\Cand_c(\CurrB, i+1)=i$.

\textbf{Case C3: $\posprevWeight[c].weight < \Weight(i)$.}
Suppose that neither \textbf{C1} nor \textbf{C2} holds, but \textbf{C3} does.
Thus, $\Weight(i)\le\LCP[i]$ and $\posprevWeight[c].index\ge \CurrB$.
Since $\Weight(i)\le\LCP[i]$, Lemma~\ref{lem-consistency} implies that either $\BWT[j]=c$ for all $j\in[\CurrB,i)$ or $\BWT[j]\neq c$ for all such $j$.
In both cases, every index $j'\in(\CurrB,i)$ with $\BWT[j']=c$ satisfies
$\Weight(j')=-1$, which implies $\posprevWeight[c].index=\CurrB$.
As $\posprevWeight[c].weight=\Weight(\CurrB)$ and
$\Weight(i)>\posprevWeight[c].weight$ by \textbf{C3}, we have
$\Cand_c(\CurrB, i+1)=i$.
%
\end{proof}

\subsection{The Complexity of the Adjusted Algorithm}
\label{sect-complexity}

We first analyze the working space.
By Proposition~\ref{prop}, the monotone stack over $\LCP[1..n]$ contains $\mathcal O(h)$
triples and thus uses $\mathcal O(h)$ words of space, where $h$ is the height of the suffix
tree built over the reversed text.
The doubly-linked list $\RowList$, together with the arrays
$\Map[1..\sigma]$ and $\posprevWeight[1..\sigma]$, requires $\mathcal O(\sigma)$ space.
Finally, the lists $\Result_c$ for $c \in \Sigma$ store at most $\chi \le n$
positions in total, where $\chi$ is the size of the suffixient array.
Hence, excluding the $\SA$, $\LCP$, and $\BWT$ arrays, the overall working space is
$\mathcal O(\chi + \sigma + h) = \mathcal O(\chi + h)$, since $\sigma \le \chi$.

We now analyze the running time.
The algorithm performs a single left-to-right scan of the arrays
$\SA$, $\LCP$, and $\BWT$, executing $n$ iterations in total.
The operations $\Weight(\cdot)$ and $\Begin(\cdot)$ are invoked at most twice per
iteration, yielding $\mathcal O(1)$ amortized time per iteration.

By Lemma~\ref{observe-2}, the triples in $\RowList$ are maintained in
non-increasing order of their weights.
Thus, in {Step 1} of the $i$-th iteration, each triple
$(p_{\text{text}}, c, w)$ with $w > \LCP[i]$ can be identified and
removed from the head of $\RowList$ in constant time.

In the second step, at most one triple is removed from $\RowList$ and at most one
new triple is inserted at its front.
Using the pointer stored in $\Map[\BWT[i]]$, the triple to be removed can be
located in constant time.
Since at most one triple is added per iteration, at most $n$ triples are added to
and removed from $\RowList$ over the entire execution.

After all $n$ iterations, the remaining triples in $\RowList$, if any, are
enumerated in $\mathcal O(\sigma)$ time.
Concatenating the lists $\Result_1, \dots, \Result_\sigma$ to produce the output
takes an additional $\mathcal O(\sigma)$ time.
Overall, the total running time is $\mathcal O(n + \sigma) = \mathcal O(n)$, assuming $\sigma \le n$.

Combining Lemmas~\ref{lem-result_c-to-L_c}, \ref{lem-remove-from-rowlist}, and
\ref{lem-order} with the analysis above, we obtain the following result.

\begin{theorem}\label{theorem}
By scanning the arrays $\BWT[1..n]$, $\SA[1..n]$, and $\LCP[1..n]$ of the reversed
input text $T[1..n]$ over an alphabet of size $\sigma$ in a single pass, one can
construct a suffixient array of $T[1..n]$ in $\mathcal O(n)$ time using $\mathcal O(\chi + h)$ words
of working space in the worst case, in addition to these arrays, where $\chi$
denotes the size of the suffixient array and $h$ denotes the height of the suffix
tree built over the reversed text.
\end{theorem}


\bibliography{lipics-v2021-sample-article}

\appendix

\newpage

\section{The Pseudocode for Computing $\Weight(i)$}
\label{app-pseudo-weight}

Let ${\mathtt{pre\_bwt}}:=\BWT[x-1]$ if $x>1$ and $-1$ otherwise, ${\mathtt{curr\_bwt}}:=\BWT[x]$, and ${\mathtt{next\_bwt}}:=\BWT[x+1]$ if $x<n$ and $-1$ otherwise.
Similarly, let ${\mathtt{curr\_lcp}}:=\LCP[x]$, and ${\mathtt{next\_lcp}}:=\LCP[x+1]$ if $x<n$, and $-1$ otherwise.
The procedure \texttt{Compute\_$\Weight$} determines $\Weight(x)$ as Algorithm \ref{alg-compute-weight}.
By scanning the $\LCP[1..n]$ and $\BWT[1..n]$ arrays in a single pass, we can apply Algorithm~\ref{alg-compute-weight} to compute $\Weight(x)$ for all $x=1, 2, \dots$.

\begin{algorithm}
\caption{\label{alg-compute-weight} \texttt{Compute}\_$\Weight({\mathtt{pre\_bwt}}, {\mathtt{curr\_bwt}}, {\mathtt{next\_bwt}}, {\mathtt{curr\_lcp}}, {\mathtt{next\_lcp}})$}
\begin{tabbing}
\noindent
\hspace{1cm}\= \hspace{1cm}\= \hspace{1cm}\= \kill
01. \> {is\_start} $\leftarrow ({\mathtt{pre\_bwt}} \neq -1)$ \textbf{and} $({\mathtt{curr\_bwt}} \neq {\mathtt{pre\_bwt}})$; \\
02. \> {is\_end} $\leftarrow ({\mathtt{next\_bwt}} \neq -1)$ \textbf{and} $({\mathtt{curr\_bwt}} \neq {\mathtt{next\_bwt}})$; \\
03. \> \textbf{if} {is\_start} \textbf{then}\\
04. \> \> $w_{start} \leftarrow {\mathtt{curr\_lcp}}$;\\
05. \> \textbf{else}\\
06. \> \> $w_{start} \leftarrow -1$;\\
03. \> \textbf{if} {is\_end} \textbf{then}\\
04. \> \> $w_{end} \leftarrow {\mathtt{next\_lcp}}$;\\
05. \> \textbf{else}\\
06. \> \> $w_{end} \leftarrow -1$;\\
07. \> \textbf{return} $\max(w_{start}, w_{end});$
\end{tabbing}
\vspace{-\baselineskip}
\end{algorithm}

\section{The Pseudocode for Computing $\Begin(i)$}
\label{app-pseudo-begin}

\begin{algorithm}
\caption{\texttt{Initial-Stack}()}
\begin{tabbing}
\hspace{1cm}\= \hspace{1cm}\= \hspace{1cm}\= \hspace{1cm}\= \kill
01. \> $S \leftarrow$ an empty stack \\
02. \> push $(1, -1, -1)$ as $(index, val, b\_val)$ onto $S$ \\
03. \> \textbf{return} $S$
\end{tabbing}
\vspace{-\baselineskip}
\end{algorithm}

\begin{algorithm}
\caption{\texttt{Compute\_b}($S, \LCP[i], i$)}
\begin{tabbing}
\hspace{1cm}\= \hspace{1cm}\= \hspace{1cm}\= \hspace{1cm}\= \kill
01. \> $tuple \leftarrow S.top()$ \\
02. \> \textbf{if} $tuple.index = i$ \textbf{then} \\
03. \> \> \textbf{return} $tuple.b\_val$;\\ 
04. \> \textbf{while} $tuple.val \ge \LCP[i]$ \textbf{do} \\
05. \> \> $S.pop()$ \\
06. \> \> $tuple \leftarrow S.top()$ \\
07. \> \textbf{end while} \\
08. \> push $(i, \LCP[i], tuple.index)$ onto $S$ \\
09. \> \textbf{return} $tuple.index$
\end{tabbing}
\vspace{-\baselineskip}
\end{algorithm}

\section{The Pseudocode and an Example for Computing Suffixient Arrays}
\label{app-pseudo-sa}

\begin{algorithm}
\caption{\label{alg-sA} \texttt{Compute Suffixient Arrays}($\BWT[1..n], \SA[1..n], \LCP[1..n]$)}
\begin{tabbing}
\hspace{0.8cm}\= \hspace{0.8cm}\= \hspace{0.8cm}\= 
\hspace{0.8cm}\= \hspace{0.8cm}\= \hspace{0.8cm}\=\hspace{7cm}\=\kill
01. \> $\BWT[n+1]\leftarrow \LCP[n+1] \leftarrow -1$;\\
02. \> $\Map[1..\sigma] \leftarrow \{\}$; \\
03. \> $\RowList \leftarrow$ doubly-linked list; \\
04.\> $S \leftarrow$ \texttt{Initial-Stack}(); \\
05. \> \textbf{for} $i \leftarrow 1$ \textbf{to} $\sigma$ \textbf{do} \\
06. \> \> $\Result_i \leftarrow$ an empty list; \\
07. \> \textbf{if} $\BWT[2]\ne \BWT[1]$ \textbf{then} \\
08. \> \> $c \leftarrow BWT[1]$; \\
09. \> \> append $(n - \SA[1] + 1, \BWT[1], \LCP[2])$ as $(p_{\text{text}}, char, weight)$ to the front of $\RowList$; \\
10. \> \> $\Map[c] \leftarrow$ head of $\RowList$; \\
11. \> \textbf{for} $i \leftarrow 2$ \textbf{to} $n$ \textbf{do} \\
12. \> \> $\ell \leftarrow \LCP[i]$; \> \> \> \> \> $\triangleright$ \textit{Step 1 starts}\\
13. \> \> $headTriple \leftarrow \RowList.head$; \\
14. \> \> \textbf{while} $\ell < headTriple.weight$ \textbf{do} \\
15. \> \> \> append $headTriple.p_{\text{text}}$ to $\Result_{headTriple.char}$;  \\
16. \> \> \> remove the head entry from $\RowList$; \\
17. \> \> \> $\Map[headTriple.char] \leftarrow \texttt{null}$; \\
18. \> \> $currWeight \leftarrow$ Compute\_$\Weight(\BWT[i-1], \BWT[i], \BWT[i+1], \LCP[i], \LCP[i+1])$; \\
19. \> \> $curr\_b \leftarrow$ Compute\_b$(S, \LCP[i], i)$; \> \> \> \> \>$\triangleright$ \textit{Step 1 ends} \\
20. \> \> \textbf{if} $currWeight > -1$ \textbf{then} \> \> \> \> \> $\triangleright$ \textit{Step 2 starts}\\
21. \> \> \> \textbf{if} $currWeight\ne \LCP[i]$ \textbf{then} \\
22. \> \> \> \> $curr\_b \leftarrow$ Compute\_b$(S, \LCP[i+1], i+1)$; \\
23. \> \> \> $c \leftarrow \BWT[i]$; \\
24. \> \> \> \textbf{if} $i=\Cand_c(curr\_b, i+1)$ \textbf{then} \\
25. \> \> \> \> \textbf{if} $\Map[c] \ne null$ \textbf{then} \\
26. \> \> \> \> \>$triple \leftarrow \Map[c]$; \\
27. \> \> \> \> \>   remove $triple$ from $\RowList$; \\
28. \> \> \> \>  append $(n - \SA[i] + 1, \BWT[i], currWeight)$ to the front of $\RowList$; \\
29. \> \> \> \>  $\Map[c] \leftarrow$ head of $\RowList$; \>\>\> $\triangleright$ \textit{Step 2 ends} \\
30. \> \textbf{while} $\RowList\ne \texttt{null}$ \textbf{do} \>\>\>\>\> \>$\triangleright$ \textit{Final step starts} \\
31. \> \> $headTriple \leftarrow \RowList.head$; \\
32. \> \> append $headTriple.p_{\text{text}}$ to $\Result_{headTriple.char}$;  \\
33. \> \> remove the head entry from $\RowList$; \\
34. \>  \textbf{return} $[\Result_1, \Result_2, \dots, \Result_{\sigma}]$; \>\>\>\>\> \>$\triangleright$ \textit{Final step ends}
\end{tabbing}
\vspace{-\baselineskip}
\end{algorithm}

\begin{sidewaystable}[p] 
\centering
\caption{\label{table-alg} Execution trace of the algorithm computing  a suffixient array for the input $T = \text{AGCACAGCA}\$$. The table provides  $T^{rev}[1..n]$, the $\LCP, \BWT$, and $\SA$ arrays, the variable $curr\_b$, the contents of $\RowList$, and the character-specific result lists $Result_c$. The final algorithm output is $[10, 1, 5, 7]$ as the suffixient array.}
\scriptsize 
\renewcommand{\arraystretch}{1.5} 
\setlength{\tabcolsep}{4pt} 

\begin{tabular}{|c|c|c|c|c|c|c|c|c|c|c|c|}
\hline
\textbf{$i$} & \textbf{1} & \textbf{2} & \textbf{3} & \textbf{4} & \textbf{5} & \textbf{6} & \textbf{7} & \textbf{8} & \textbf{9} & \textbf{10} & \textbf{n+1} \\ \hline
$T[i]$ & A & G & C & A & C & A & G & C & A & \$ & --- \\ \hline
$T^{rev}[i]$ & A & C & G & A & C & A & C & G & A & \$ & --- \\ \hline
$\LCP[i]$ & -1 & 0 & 1 & 2 & 4 & 0 & 1 & 3 & 0 & 2 & --- \\ \hline
$\BWT[i]$ & A & G & G & C & \$ & A & A & A & C & C & --- \\ \hline
$\Weight(i)$ & 0 & 0 & 2 & 4 & 4 & 0 & -1 & 0 & 0 & -1 & --- \\ \hline
$\SA[i]$ & 10 & 9 & 4 & 6 & 1 & 5 & 7 & 2 & 8 & 3 & --- \\ \hline
$n-\SA[i]+1$ & 1 & 2 & 7 & 5 & 10 & 6 & 4 & 9 & 3 & 8 & --- \\ \hline
$\LCP[i]=\Weight(i)?$ & F & T & F & F & T & T & F & F & T & F & --- \\ \hline
$curr\_b$ & 1 & 1 & 3 & 4 & 4 & 1 & 6 & 1 & 1 & 9 & --- \\ \hline
$\Weight(i) > -1?$ & T & T & T & T & T & T & F & T & T & F & --- \\ \hline
$i = \Cand_c(curr\_b, i+1)$ & T & T & T & T & T & F & N/A & F & F & N/A & --- \\ \hline
$\RowList$ & (1,A,0) & (2,G,0) & (7,G,2) & (5,C,4) & (10,\$,4) & (1,A,0) & (1,A,0) & (1,A,0) & (1,A,0) & (1,A,0) & --- \\ 
& & (1,A,0) & (1,A,0) & (7,G,2) & (5,C,4) & & & & & & \\ 
& & & & (1,A,0) & (7,G,2) & & & & & & \\ 
& & & & & (1,A,0) & & & & & & \\ \hline
$\Result_{\$}$ & & & & & 10 & 10 & 10 & 10 & 10 & 10 & 10 \\ \hline
$\Result_{A}$ & & & & & &  &  &  &  &  & 1 \\ \hline
$\Result_{C}$ & & & & & 5 & 5 & 5 & 5 & 5 & 5 & 5 \\ \hline
$\Result_G$ & & & & & 7 & 7 & 7 & 7 & 7 & 7 & 7 \\ \hline
\end{tabular}
\end{sidewaystable}

\newpage

\end{document}